\documentclass[submission,copyright,creativecommons]{eptcs}
\providecommand{\event}{TARK 2017} 
\usepackage{underscore}           
\usepackage{verbatim}
\usepackage{textcomp}
\usepackage{amstext}
\usepackage{amsthm}
\usepackage{amssymb}
\usepackage{stmaryrd}

\newtheorem{theorem}{Theorem}
  \theoremstyle{example}
  \newtheorem{example}[theorem]{Example}
  \theoremstyle{definition}
  \newtheorem{definition}[theorem]{Definition}
  \theoremstyle{proposition}
  \newtheorem{proposition}[theorem]{Proposition}
  \theoremstyle{lemma}
  \newtheorem{lemma}[theorem]{Lemma}

\newenvironment{lyxlist}[1]
{\begin{list}{}
{\settowidth{\labelwidth}{#1}
 \setlength{\leftmargin}{\labelwidth}
 \addtolength{\leftmargin}{\labelsep}
 }}
{\end{list}}

\newcommand{\commentout}[1]{}

\newcommand{\Rat}{\mathbb{Q}}
\newcommand{\Ag}{{\text{\it Ag}}}
\newcommand{\Wei}{\mathbb{L}}

\usepackage{tikz}
\tikzstyle{aworld}=[circle,draw=black!50,fill=black!20,thick]
\tikzstyle{world}=[circle,draw=black!50,thick]
\tikzstyle{wworld}=[circle,draw=black!55,thick,circle,draw=black!50,thick]
\tikzstyle{aaction}=[rectangle,draw=black!50,fill=black!20,thick]
\tikzstyle{action}=[rectangle,draw=black!50,thick]

\usetikzlibrary{automata,arrows}

\title{Conditional Belief, Knowledge and Probability}

\author{Jan van Eijck
  \institute{CWI and ILLC\\ Amsterdam, Netherlands}
  \email{jan.van.eijck@cwi.nl}
\and
Kai Li\thanks{Kai Li thanks the China Scholarship Council (CSC) for a visiting grant
to CWI, and CWI for hospitality during academic year 2016/2017. 
We wish to thank Dick de Jongh for useful advice, and Malvin Gattinger for his help and advice when we were organizing ideas of this paper. We are also grateful to 3 anonymous reviewers for their comments on an earlier version of the paper.}
  \institute{Peking University\\ Beijing, China}
 \institute{CWI\\ Amsterdam, Netherlands}
  \email{likaiedemon@gmail.com}
}
\def\titlerunning{Conditional Belief, Knowledge and Probability}
\def\authorrunning{Jan van Eijck \& Kai Li}
\begin{document}
\maketitle
\begin{abstract}
  A natural way to represent beliefs and the process of updating
  beliefs is presented by Bayesian probability theory, where belief of 
  an agent $a$ in $P$ can be interpreted as $a$ considering that $P$ is more
  probable than not $P$. This paper attempts to get at the core logical 
  notion underlying this. 

  The paper presents a sound and complete neighbourhood logic for
  conditional belief and knowledge, and traces the connections with
  probabilistic logics of belief and knowledge. The key notion in this
  paper is that of an agent $a$ believing $P$ conditionally on having
  information $Q$, where it is assumed that $Q$ is compatible with
  what $a$ knows.  

  Conditional neighbourhood logic can be viewed as a core system for
  reasoning about subjective plausibility that is not yet committed to
  an interpretation in terms of numerical probability. Indeed, every
  weighted Kripke model gives rise to a conditional neighbourhood
  model, but not vice versa.  We show that our calculus for
  conditional neighbourhood logic is sound but not complete for
  weighted Kripke models.  Next, we show how to extend the calculus
  to get completeness for the class of weighted Kripke models.

  Neighbourhood models for conditional belief are closed under model
  restriction (public announcement update), while earlier
  neighbourhood models for belief as `willingness to bet' were
  not. Therefore the logic we present improves on earlier
  neighbourhood logics for belief and knowledge. We present complete 
  calculi for public announcement and for publicly revealing the 
  truth value of propositions using reduction axioms. The reductions show 
  that adding these announcement operators to the language does not increase 
  expressive power. 
\end{abstract}


\section{Introduction}\label{section:int}

This paper aims at isolating a core logic of rational belief and
belief update that is compatible with the Bayesian picture of rational
inference \cite{Jaynes:pttlos}, but that is more general, in the
sense that it does not force epistemic weight models (that is, models
with fixed subjective probability measures, for each agent) on us.

Epistemic neighbourhood models, as defined in \cite{EijckRenne2016:upkb},
represent belief as truth in a neighbourhood, where the neighbourhoods
for an agent are subsets of the current knowledge cell of that agent.
Intuitively, a neighbourhood lists those propositions compatible with
what the agent knows that the agent considers as more likely than
their complements. Since we intend to use neighbourhood semantics to
model a certain kind of belief, it is natural to study belief
updates. However, there is an annoying obstacle for updates even in a
very simple case: public announcement.

Intuitively, a public announcement would make an agent restrict
his/her belief to the announced case. A natural way to implement this
is by restricting every belief-neighbourhood to $\phi$-worlds after
announcing $\phi$. The following example shows that this does not
work, because this kind of update does not preserve reasonable
neighbourhood conditions.  Suppose Alice, somewhat irrationally,
believes that ticket $t$ she has bought is the winning ticket in a
lottery that she knows has 10,000 tickets.  Let $n$ represent the
world where ticket $n$ is winning (we assume that this is a single
winner lottery, and that Alice knows this).  Then Alice's belief is
given by a neighbourhood model with
\[
  N_a (n) = \{ X \subseteq \{0000,\ldots,9999\} \mid t \in X \},
\]
for all $n \in \{ 0000, \ldots, 9999 \}$.  Note that Alice's belief
does not depend on the world she is in: if $n,m$ are different worlds,
then $N_a (n) = N_a (m)$.

Assume Alice gets the information that some ticket $v$ different from
the ticket $t$ that she bought is the winning ticket. Let $p$ be such
that $V(p) = \{ v \}$.  Then updating with $p$ leads to an updated
model with world set $\{v\}$, and with
$N'_{a}(v)=\{X\cap\{v\}\mid X\in N_{a}(v)\}=\{\emptyset,\{v\}\}$.
However, $N'_{a}(v)$ is not a neighbourhood function, for it
contradicts the condition that beliefs should be consistent (different
from $\emptyset$).

The rest of the paper is structured as follows. Section
\ref{section:cns} introduces conditional neighbourhood semantics as an
enrichment of epistemic neighbourhood semantics, and presents a
complete calculus for it.  In Section \ref{section:wei} we show that
this calculus is sound but not complete for epistemic weight models,
and next, that our language is expressive enough to allow an extension
to a complete system for weight models. Section \ref{section:pub}
shows that conditional neighborhood models are an excellent starting
point for an extension with public announcement update
operators. Section \ref{section:cafw} traces connections with the
literature, lists further work, and concludes.

\section{Conditional Neighbourhood Semantics} \label{section:cns}

Epistemic neighbourhood models are defined in
\cite{EijckRenne2016:upkb} as epistemic models
$\mathfrak{M}=(W,\sim,V)$ with a neighbourhood function $N$ added to
them. The neighbourhood function assigns to each agent $a$ and each
world $w\in W$ a neighbourhood $N_a(w)$ that consists of the set of
propositions that agent $a$ believes in $w$.

Intuitively each element in neighbourhood $N_a(w)$ represents a belief
agent $a$ holds. Usually belief is bolder than knowledge. Indeed, most
people believe many things of which they are not sure. If we equate
certainty with knowledge, then this means that any belief of agent $a$
should be a subset of agent $a$'s current \emph{knowledge cell}, i.e.,
the set $[w]_a=\{u\in W\mid w\sim_au\}$. It follows that each
proposition in neighbourhood $N_a(w)$ is a subset of $[w]_a$.  Thus it
is natural (in the framework of epistemic modal logic) to assume the
following neighbourhood conditions:
\begin{description}
\item [Monotonicity] If an agent believes $X$, and knows that $X$
  entails $Y$, then the agent believes $Y$.
\item [No-inconsistency] An agent does not hold an inconsistent belief.
\item [Determinacy] An agent do not believe both a proposition and its
  complement.
\end{description}

However as is illustrated by Alice's Lottery example from the introductory
section, public announcement update do not preserve the 
\textbf{No-inconsistency} condition.  In order to overcome this
problem, we propose to enrich neighbourhood functions $N$ with an
extra parameter for propositions. In other words, instead of focusing on what
agents believe, we turn our attention to what agents would believe
under some assumption. Following this intuition, for each proposition
$X$, $N_{a}^{w}(X)$ is a set of propositions such that each of these
propositions represents a belief agent $a$ holds at state $w$ when
supposing $X$. In this paper, we are interested in beliefs as
`willingness to bet', i.e., an agent believes a proposition $Y$
supposing $X$ if the agent considers $Y\cap X$ more likely to be
true than its complement conditioned by $X$, namely $-Y\cap X$.  We
also assume the following postulate:
\begin{description}
\item [Equivalence of Conditions] If an agent knows that two conditions are
equivalent, then the agent's beliefs are the same under both conditions.
\end{description}
For non-equivalent conditions, on the other hand, conditional beliefs
may vary.

Assume $p$ ranges over a set of proposition letters $P$, and $a$ over
a finite set of agents $A$.  The language for conditional
neighbourhood logic ${\mathcal L}_{CN}$ is given by the following BNF
definition:
\[
\phi::=\top\mid p\mid\neg\phi\mid\left(\phi\wedge\phi\right)\mid B_{a}(\phi,\phi)
\]
$B_{a}(\phi,\psi)$ can be read as ``assuming $\phi$, agent $a$
believes (is willing to bet) $\psi$''. $\perp,\vee,\to,\leftrightarrow$ are defined as
usual. Somewhat arbitrarily, we assume that conditioning with
information that contradicts what the agent knows (is certain of) will
cause an agent to believe nothing anymore. This means we can define
knowledge in terms of conditional belief, as follows. Use $K_{a}\phi$
for $\neg B_{a}(\neg\phi,\top)$ (which can be read as ``$\neg\phi$
contradicts what agent $a$ is certain of'',) and
$\check{K}_{a}\phi$ for $\neg K_{a}\neg\phi$.

Consider Alice's lottery situation again.  Alice knows there are
10,000 lottery tickets numbered 0000 through 9999.  Alice believes
ticket $t$ is winning (and buys it).  Let $n$ represent the world
where ticket $n$ is winning. Then Alice's belief is given by a
conditional neighbourhood model with
$N^w_a (X) = \{ Y \subseteq X
\mid t \in Y \}$ if
$t \in X$, and
$N^w_a (X) = \{ Y \subseteq X \mid |Y| > \frac 1 2 | X | \}$ if
$t \notin X$.  Now Alice receives the information that $v \neq t$ is
the winning ticket.  Then $v = w$, the updated model has universe
$\{ v \}$, and Alice updates her belief to $N'$ with
${N'}^v_a (\{ v \}) = \{ \{ v \} \}$. In the updated model, Alice
knows that $v$ is the winning ticket.

\begin{definition}
\label{def:NeighborModel} Let $\Ag$ be a finite set of agents. A
conditional neighbourhood model ${\mathfrak M}$ is a tuple $\left(W,N,V\right)$
where:
\begin{itemize}
\item $W$ is a non-empty set of worlds;
\item $N:\Ag\times W\times\mathcal{P}W\to\mathcal{P}\mathcal{P}W$ is a function that assigns to every agent $a\in\Ag$, every world $w\in W$ and
  set of worlds $X\subseteq W$ a collection $N_{a}^{w}(X)$ of sets of
  worlds\textendash each such set called a neighbourhood of
  $X$\textendash subject to the following conditions, where
\[
   [w]_{a}=\{v\in W\mid\forall X\subseteq W,N_{a}^{w}(X)=N_{a}^{v}(X)\}: 
\]
\begin{description}
\item [{(c)}] $\forall Y\in N_{a}^{w}(X):Y\subseteq X\cap[w]_{a}$.
\item [{(ec)}] $\forall Y\subseteq W$: if $X\cap[w]_{a}=Y\cap[w]_{a}$,
then $N_{a}^{w}(X)=N_{a}^{w}(Y)$.
\item [{(d)}] $\forall Y\in N_{a}^{w}(X),X\cap[w]_{a}-Y\notin N_{a}^{w}(X)$.
\item [{(sc)}] $\forall Y,Z\subseteq X\cap[w]_{a}:$ if $X\cap[w]_{a}-Y\notin N_{a}^{w}(X)$
and $Y\subsetneq Z$, then $Z\in N_{a}^{w}(X)$.
\end{description}
\item $V$ is a valuation.
\end{itemize}
\end{definition}
We call $N$ a neighbourhood function; a neighbourhood $N_{a}^{w}(X)$ for agent
$a$ in $w$, conditioned by $X$ is a set of propositions each of which
agent $a$ believes more likely to be true than its complement.

Property (c) expresses that what is believed is also known; (ec) expresses
\textbf{equivalence of conditions}; (d) expresses \textbf{determinacy}; 
(sc) expresses a form of ``strong commitment'': if the agent does not 
believe the complement of $Y$ then she must believe any weaker $Z$ implied
by $Y$. It can be proved (see Appendix \ref{appendix:alternativeDef}, Lemma \ref{lem:implied-conditions}) that
these conditions together imply that any conditional neighbourhood model
 ${\mathfrak M} = (W,N,V)$ also satisfies the following, for any $a \in A$, 
$w \in W$, $X \subseteq W$: 
\begin{description}
\item [{(m)}] $\forall Y\subseteq Z\subseteq X\cap[w]_{a}:$ if $Y\in N_{a}^{w}(X)$,
then $Z\in N_{a}^{w}(X)$;
\item [{(ni)}] $\emptyset\notin N_{a}^{w}(X)$;
\item [{(n){*}}] if $X\cap[w]_{a}\neq\emptyset$, then $X\cap[w]_{a}\in N_{a}^{w}(X)$;
\item [{($\emptyset$)}] if $X\cap[w]_{a}=\emptyset$, then
  $N_{a}^{w}(X)=\emptyset$;
\end{description}
where (m) and (ni) expresses \textbf{monotonicity} and \textbf{no-inconsistency} respectively. ($\emptyset$) expresses that conditioning with
information that contradicts what the agent knows will
cause an agent to believe nothing anymore. Note that ($\emptyset$) reflects our definition for $K$-operators $K_a\phi::=\neg B_a(\neg\phi,\top)$.

Let ${\mathfrak M}=\left(W,N,V\right)$ be a conditional neighbourhood
model, let $w\in W$. Then the key clause of the truth definition is 
given by: 
\[
\begin{array}{lcl}
{\mathfrak M},w\vDash B_{a}(\phi,\psi) & \text{iff} & \text{for some }Y\in N_{a}^{w}(\left\llbracket \phi\right\rrbracket _{{\mathfrak M}})\text{, }Y\subseteq\left\llbracket \psi\right\rrbracket _{{\mathfrak M}}
\end{array}
\]
where $\left\llbracket \phi\right\rrbracket _{{\mathfrak M}}=\{w\in W\mid{\mathfrak M},w\vDash\phi\}$.
Because of (m), we can prove that
\[
{\mathfrak M},w\vDash B_{a}(\phi,\psi)\text{ iff }\{v\in\left\llbracket \phi\right\rrbracket _{{\mathfrak M}}\cap[w]_{a}\mid{\mathfrak M},v\vDash\psi\}\in N_{a}^{w}(\left\llbracket \phi\right\rrbracket _{{\mathfrak M}}).
\]
It is worth noting that by (ni), $B_a(\phi,\bot)$ will always be invalid for any agent $a$ and any formula $\phi$.

Note that conditional neighborhood models do not have epistemic
relations $\sim_a$. However, such relations can be introduced by the
neighourhood function as follows: for each $a\in Ag$,
$\sim_{a}\ \subseteq W\times W$ is a relation such that
\[
\forall w,v\in W,w\sim_{a}v \text{ iff }\forall X\subseteq W,\
  N_{a}^{w}(X)=N_{a}^{v}(X).
\]
Then
${\mathfrak M},w\vDash K_{a}\phi\text{ iff for each }v\sim_{a}w\text{,
}{\mathfrak M},v\vDash\phi$. 

It can be proved that the version of conditional neighbourhood models
with epistemic relations $\sim_a$ is equivalent to the version without
(see Appendix \ref{appendix:alternativeDef}). Such equivalence is guaranteed by properties (n)*, ($\emptyset$) and another property which was can be found in \cite{DBLP:conf/atal/BalbianiDL16}(, however they use neighbourhoods instead of conditional neighbourhoods for beliefs):
\begin{lyxlist}{00.00.0000}
\item [{(a)}] $\forall v\in\left[w\right]_{a}:N_{a}^{w}(X)=N_{a}^{v}(X)$,
\end{lyxlist}
which states that if a agent cannot distinguish two worlds, then the agent holds the same beliefs on either of these two worlds. 
Thus we do
not differentiate conditional neighbourhood models with or without
such relations.

\begin{figure}[htbp]
\begin{lyxlist}{00.00.0000}
\item [{(Taut)}] \begin{flushleft}
All instances of propositional tautologies
\par\end{flushleft}
\item [{(Dist-K)}] \begin{flushleft}
$K_{a}(\phi\to\psi)\to K_{a}\phi\to K_{a}\psi$
\par\end{flushleft}
\item [{(T)}] \begin{flushleft}
$K_{a}\phi\to\phi$
\par\end{flushleft}
\item [{(5B)}] \begin{flushleft}
$B_{a}(\phi,\psi)\to K_{a}B_{a}(\phi,\psi)$
\par\end{flushleft}
\item [{(4B)}] \begin{flushleft}
$\neg B_{a}(\phi,\psi)\to K_{a}\neg B_{a}(\phi,\psi)$
\par\end{flushleft}
\item [{(D)}] \begin{flushleft}
$B_a(\phi, \psi) \rightarrow \neg B_a (\phi, \neg \psi)$
\par\end{flushleft}
\item [{(EC)}] \begin{flushleft}
$K_{a}(\phi\leftrightarrow\psi)\to B_{a}(\phi,\chi)\to B_{a}(\psi,\chi)$ 
\par\end{flushleft}
\item [{(M)}] \begin{flushleft}
$K_{a}(\phi\to\psi)\to B_{a}(\chi,\phi)\to B_{a}(\chi,\psi)$ 
\par\end{flushleft}
\item [{(C)}] $B_{a}(\phi,\psi)\to B_{a}(\phi,\phi\wedge\psi)$
\item [{(SC)}] 	$\check{B}_{a}\phi\land\check{K}_{a}(\lnot\phi\land\psi)\to B_{a}(\phi\lor\psi)$
\end{lyxlist}
Rules: 
\[
\begin{array}{c}
\phi\to\psi\quad\phi\\
\hline \psi
\end{array}\;\text{\footnotesize(MP)}\qquad\begin{array}{c}
\phi\\
\hline K_{a}\phi
\end{array}\;\text{\footnotesize(Nec-K)}
\]
\caption{The CN Calculus for Conditional Neighbourhood Logic} \label{figure:CNcalculus}
\end{figure}

In Figure \ref{figure:CNcalculus}, axiom (D) guarantees the truth of neighbourhood condition (d), (EC) would correspond to (ec), (M) to (ec), (M) to (m), (C) to (c) and (SC) to (sc).

\begin{theorem} \label{thm:CompCN}
The CN calculus for Conditional
Neighbourhood logic given in Figure \ref{figure:CNcalculus} is sound
and complete for conditional neighbourhood models. 
\end{theorem} 
\begin{proof}
See Appendix \ref{appendix:CompCNproof}.
\end{proof}
Note that the
calculus does not have $4$ and $5$ for $K$; this is because these
principles are derivable from (5B) and (4B).

For an example of an interesting principle that can already be proved
in the CN calculus, consider the following. Suppose we have a biased
coin with unknown bias, and we want to use it to simulate a fair coin.
Then we can use a recipe first proposed by John von Neumann
\cite{Neumann1951:vt}: toss the coin twice. If the two outcomes are
not the same, use the first result; if not, forget the outcomes and
repeat the procedure. Why does this work? Because we can assume that
two tosses of the same coin have the same likelihood of showing heads,
even if the coin is biased. We can express subjective likelihood
comparison in our language.  See Figure \ref{fig:Comparison}, where we
use $\alpha$ for ``the first toss comes up with heads'', and $\beta$
for ``the second toss comes up with heads''. This hinges on the
following principle:
\begin{quote}
  If $\alpha$ and $\beta$ have the same likelihood, then 
   $\alpha \land \neg \beta$ and $\neg \alpha \land \beta$ 
  should also have the same likelihood, and vice versa. 
\end{quote}
Notice that $B_{a}(\alpha\leftrightarrow\neg\beta,\alpha)$ expresses
that agent $a$ considers $\alpha \land \neg \beta$ more likely than
$\neg \alpha \land \beta$. And it follows from the comparison principle in Figure \ref{fig:Comparison} that this 
is equivalent to: $a$ considers $\alpha$ more likely than $\beta$. From 
now on, let $\alpha \succ_{a}\beta$ abbreviate  $B_{a}(\alpha\leftrightarrow\neg\beta,\alpha)$. 

By axioms $M$ and $C$, $B_a( \alpha \leftrightarrow \neg \beta, \beta) 
\leftrightarrow B_a( \alpha \leftrightarrow \neg \beta, \neg \alpha)$, and 
by axiom $D$, $B_a( \alpha \leftrightarrow \neg \beta, \alpha) \rightarrow \neg B_a( \alpha \leftrightarrow \neg \beta, \neg \alpha)$. Therefore,  
$B_a( \alpha \leftrightarrow \neg \beta, \alpha) \rightarrow \neg B_a( \alpha \leftrightarrow \neg \beta, \beta)$ is provable in the CN calculus. Using the abbreviation: $\alpha \succ_{a}\beta \rightarrow 
\neg \beta \succ_{a} \alpha$. We therefore have three mutually exclusive cases: 
\begin{itemize} 
 \item $\alpha \succ_{a}\beta$. 
 \item $\beta \succ_{a} \alpha$. 
 \item $\neg \alpha \succ_{a}\beta \land \neg \beta \succ_{a} \alpha$.
\end{itemize}
Agreeing to abbreviate the third case as $\alpha\approx_{a}\beta$, we get the following 
totality principle.
\begin{description}
\item[Totality] $(\alpha\succ_{a}\beta)\vee(\beta\succ_{a}\alpha)\vee(\alpha\approx_{a}\beta).$
\end{description}
Next, we define $\alpha\succcurlyeq_a\beta$ by $\neg\beta\succ_a\alpha$. This abbreviation gives: 
\begin{description}
\item[Refl Totality] $(\alpha\succcurlyeq_{a}\beta)\vee(\beta\succcurlyeq_{a}\alpha).$
\end{description}
We can also connect to logic languages concerning probability that do
not have knowledge operators $K_a$ but instead use
${\succcurlyeq_a}\top$ (for instance \cite{Gardenfors75} and
\cite{ghosh2012comparing}), by deriving that
$(\alpha\succcurlyeq_a\top)\leftrightarrow K_a\alpha$.

\newcommand{\boundellipse}[3]
{(#1) ellipse (#2 and #3)}

\begin{figure}
\begin{center}
\begin{tikzpicture}[scale=0.7]
   \draw \boundellipse{0,0}{3}{2};
   \draw \boundellipse{3,0}{3}{2};
   \node (a) at (0,-2.3) {$\alpha$}; 
   \node (b) at (3,-2.3) {$\beta$};
   \node (ab) at (1.5,0) {$\alpha \land \beta$}; 
   \node (a-b) at (-1.5,0) {$\alpha \land \neg\beta$}; 
   \node (b-a) at (4.5,0) {$\beta \land \neg\alpha$}; 
\end{tikzpicture}
\end{center}
\caption{\textbf{Comparison principle}: 
$\alpha \land \neg\beta \succ \beta \land \neg\alpha\ \ $ iff $\ \ \alpha \succ \beta$ 
\label{fig:Comparison}}
\end{figure}
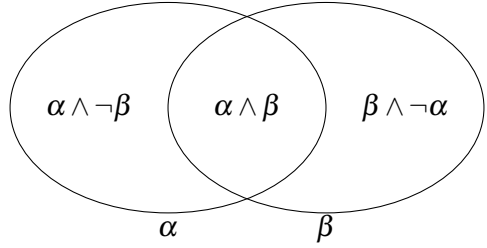

\section{Incompleteness and Completeness for Epistemic Weight Models}\label{section:wei}

In this Section we interpret ${\mathcal L}_{CN}$ in epistemic weight
models, and give an incompleteness and a completeness result.
\begin{definition}
An \emph{Epistemic Weight Model} for agents $\Ag$ and basic propositions
$P$ is a tuple $\mathfrak{M}=(W,R,L,V)$ where 
$W$ is a non-empty countable set of worlds, $R$ assigns to every
agent $a\in\Ag$ an equivalence relation $\sim_{a}$ on $W$, $L$
assigns to every $a\in\Ag$ a function $\Wei_{a}$ from $W$ to $\Rat^{+}$
(the positive rationals), subject to the following boundedness condition
({*}). 
\[
\label{*}\forall a\in\Ag\ \forall w\in W\sum_{u\in[w]_{a}}\Wei_{a}(u)<\infty.
\]
where $[w]_{a}$ is the cell of $w$ in the partition induced by $\sim_{a}$.
$V$ assigns to every $w\in W$ a subset of $P$.
\end{definition} 

We can interpret conditional belief sentences in these models.  Let
$\mathfrak{M}$ be a weight model, and let
$\left\llbracket \phi\right\rrbracket
_{\text{\ensuremath{\mathfrak{M}}}}=\{w\in
W\mid\text{\ensuremath{\mathfrak{M}}},w\vDash\phi\}$.
Let $w$ be a world of $\mathfrak{M}$. Then 
\[
\mathfrak{M},w\vDash B_a(\phi,\psi)  \text{ iff } 
\mathbb{L}_a([w]_{a}\cap\left\llbracket \phi\wedge\psi\right\rrbracket _{\text{\ensuremath{\mathfrak{M}}}})>\mathbb{L}_a([w]_{a}\cap\left\llbracket \phi\wedge\neg\psi\right\rrbracket _{\text{\ensuremath{\mathfrak{M}}}}),
\]
and $\mathfrak{M},w\vDash K_a \phi  \text{ iff  
for all } v \in [w]_a,\ \mathfrak{M},w\vDash \phi$. 
One easily checks that the axioms of the CN calculus 
are true for this interpretation, so we have: 
\begin{theorem}
The CN calculus is sound for epistemic weight models. 
\end{theorem} 

To see that we do not have completeness, observe that we can 
express Savage's {\em Sure Thing Principle} in our language. 
In Savage's example this is about action. If an agent would 
do a thing if he would know $\phi$, and would do the same 
thing if he would know $\neg \phi$, then he should do the thing
in any case: 
\begin{quote}\small
A businessman contemplates buying a certain piece of property. He
considers the outcome of the next presidential election relevant. So,
to clarify the matter to himself, he asks whether he would buy if he
knew that the Democratic candidate were going to win, and decides that
he would. Similarly, he considers whether he would buy if he knew that
the Republican candidate were going to win, and again finds that he
would. Seeing that he would buy in either event, he decides that he
should buy, even though he does not know which event obtains, or will
obtain, as we would ordinarily say. 
Savage, \cite[p 21]{Savage1972:tfosSE}.
\end{quote} 
The following formula expresses this Sure Thing Principle, not 
about action but about belief: 
\[
   B_a( \phi, \psi) \land B_a (\neg \phi, \psi) 
   \rightarrow B_a (\top, \psi). 
\]
It is not hard to see that this principle is valid for weight models:
if $\psi$ has greater weight than $\neg \psi$ within the $\phi$ area,
and also within the $\neg \phi$ area, then it is a matter of adding
these weights to see that $\psi$ has greater weight than $\neg \psi$
in the whole domain. But the Sure Thing Principle is not a validity
for neighbourhood models. Let us consider the following urn example
modified from Ellsberg's paradox \cite{Ellsberg1961:raatsa}.

\begin{example}
  An urn contains 120 balls: 30 red balls, 30 green balls, and 60
  yellow or blue balls (in some unknown proportion). A ball $x$ will
  be pulled out of this urn, and there are 3 pairs of gambles where Alice
  has to pick her choice: 
\[
\begin{array}{|l|l|}
\hline G_r: x \text{ is red} & G_y: x \text{ is yellow}\\
\hline G_g: x \text{ is green} & G_b: x \text{ is blue}\\
\hline G_{rg}: x \text{ is either red or green} & G_{yb}: x \text{ is either yellow or blue}\\
\hline\end{array}
\]
Alice knows that the likelihood of $G_{rg}$ (1/2) is the same as
$G_{yb}$ (1/2), but is uncertain of the likelihood of $G_y$ and
$G_b$. Alice is ambiguity averse in her beliefs, which means 
that she is willing to bet $G_r$ against $G_y$, and $G_g$ against $G_b$.
\end{example}

To model this example, let
$W=\{red, green, yellow, blue\}$, and let the neighbourhood function
be the same for any $w\in W$, with
\[N^w(\{red, yellow\}) = \{ \{red\}, \{red,yellow\} \},\]
\[N^w(\{green, blue\}) = \{\{green\}, \{green,blue\}\},\]
$N^w(W)$ contains all subsets with at least 3 worlds.
The reason that such model exists is that all our neighbourhood properties
do not express anything about connections between different
neighbourhoods. Thus, in this model we have $B (G_r\vee G_y, G_r)$,
$B(G_g\vee G_b, G_g)$ and $B(\top,\neg (G_r\vee G_g))$ all true in
every world, in contradiction with the Sure Thing principle.

\begin{theorem}
The CN calculus is incomplete for epistemic weight models. 
\end{theorem} 
\begin{proof}
  As we have seen, $B (G_r\vee G_y, G_r \vee G_g) \land B(G_g\vee G_b, G_r \vee G_g) \rightarrow B(\top,G_r\vee G_g)$
  is false in the above neighbourhood counterexample to Sure Thing.  So
  by Theorem \ref{thm:CompCN}, Sure Thing
  cannot be proved in the CN calculus. But Sure Thing is valid in the
  class of epistemic weighted models.
\end{proof}

\subsection*{A Complete Calculus for Epistemic Weight Models}

\cite{Segerberg1971:qpiams} and \cite{Gardenfors75} proposed a
complete logic $QP$ for epistemic weight models based on a language
$\mathcal{L}_{QP}$ given by the following BNF definition\footnote{We
  replaced ``0'' and ``1'' in \cite{Gardenfors75} with the equivalent
  notation ``$\bot$'' and ``$\top$'' respectively, and extend this
  language to multiagent case. }:
\[
\phi::=\top\mid p\mid\neg\phi\mid\left(\phi\wedge\phi\right)\mid \phi\succcurlyeq_a\phi
\]
$\succ_a$ and $\approx_a$ are given as usual. ${\succcurlyeq_a}\top$
functions as $K_a$ in epistemic modal logic. The complex formula
$\alpha_0\ldots\alpha_mE_a\beta_0\ldots\beta_m$ first proposed by
Segerberg in \cite{Segerberg1971:qpiams} is an abbreviation of the
formula expressing that for all worlds $w$ in the evaluated knowledge
cell of agent $a$, the number of $\alpha_i$ among
$\alpha_0\ldots\alpha_m$ true in $w$ is the same as those of $\beta_j$
among $\beta_0\ldots\beta_m$ true in $w$. $QP$ logic does not assume
that every world in an agent's knowledge cell has the same likelihood
for all propositions; and is with the following core axioms:
\begin{lyxlist}{00.00.0000}
\item [{(A4)}] 
$\alpha_0\alpha_1\ldots\alpha_mE_a\beta_0\beta_1\ldots\beta_m\wedge(\alpha_0\succcurlyeq_a\beta_0)\wedge\ldots\wedge(\alpha_{m-1}\succcurlyeq_a\beta_{m-1})\to(\beta_m\succcurlyeq_a\alpha_m)$ for all $m\geq 1.$
\end{lyxlist}
Nevertheless, not only we can express every notion in the
probabilistic language $\mathcal{L}_{QP}$, but using the results of
\cite{Gardenfors75} and \cite{Sco64:JMP} we can prove the following
completeness theorem as well (the proof is a simple adaptation of a proof found in \cite{Gardenfors75}, just to observe that every other axioms in \cite{Gardenfors75} except (A4) is provable in CN calculus).

\begin{theorem}
\emph{CN}$\oplus$\emph{(A4)} is complete for epistemic weight models.
\end{theorem}

\subsection*{Comparing Expressive Power}

In this subsection we compare the expressive power between
$\mathcal{L}_{CN}$ and $\mathcal{L}_{QP}$, restricting our
attention to the single-agent case. As is shown in
Section \ref{section:cns}, we can translate $\mathcal{L}_{CN}$ into
$\mathcal{L}_{QP}$ by defining
$Tr_1:\mathcal{L}_{CN}\to\mathcal{L}_{QP}$ with key case:
\[
Tr_1(B(\alpha,\beta))=Tr_1(\alpha\wedge\beta)\succ Tr_1(\alpha\wedge\neg\beta)),
\]
which express that the agent considers $\alpha\wedge\beta$ is
more likely than $\alpha\wedge\neg\beta$. Likewise we can define the
translation $Tr_2$ from $\mathcal{L}_{QP}$ to $\mathcal{L}_{CN}$ by key case
\[
Tr_2(\alpha\succcurlyeq\beta)=\neg B(Tr_2(\alpha\leftrightarrow\neg\beta),Tr_2(\beta)).
\]
It is easy to prove that both translations preserve truth on
weight models. 
As for conditional neighbourhood models, consider the
following truth definition for $\mathcal{L}_{QP}$:
\[
\mathfrak{M},w\vDash_{qp}\alpha\succcurlyeq\beta\text{ iff }\left\llbracket\neg\alpha\wedge\beta\right\rrbracket ^{qp}_{{\mathfrak M}}\notin N^w(\left\llbracket\alpha\leftrightarrow\neg\beta\right\rrbracket ^{qp}_{{\mathfrak M}}),
\]
which is equivalent to( by condition (d) and $\alpha\succ\beta$::=$\neg\beta\succcurlyeq\alpha$),
\[
\mathfrak{M},w\vDash_{qp}\alpha\succ\beta\text{ iff }\left\llbracket\alpha\wedge\neg\beta\right\rrbracket ^{qp}_{{\mathfrak M}}\in N^w(\left\llbracket\alpha\leftrightarrow\neg\beta\right\rrbracket ^{qp}_{{\mathfrak M}}).
\]

Such truth condition parallels with our translation $Tr_2$. Furthermore, to show that $Tr_1$ preserves truth value on conditional neighbourhood models, we can prove by induction on construction of $\mathcal{L}_{CN}$-formulas, and establish the following equivalences:
\[
\begin{array}{ccl}
\mathfrak{M},w\vDash_{qp} Tr_1(B(\alpha,\beta)) & \text{ iff } & \mathfrak{M},w\vDash_{qp}\alpha'\wedge\beta'\succ\alpha'\wedge\neg\beta'\\
& \text{ iff } & \left\llbracket(\alpha'\wedge\beta')\wedge\neg(\alpha'\wedge\neg\beta')\right\rrbracket ^{qp}_{{\mathfrak M}}\in N^w(\left\llbracket(\alpha'\wedge\beta')\leftrightarrow\neg(\alpha'\wedge\neg\beta')\right\rrbracket ^{qp}_{{\mathfrak M}})\\
& \text{ iff } & \left\llbracket\alpha'\wedge\beta'\right\rrbracket ^{qp}_{{\mathfrak M}}\in N^w(\left\llbracket\alpha'\right\rrbracket ^{qp}_{{\mathfrak M}})\\
& \text{ iff } & \left\llbracket\beta'\right\rrbracket ^{qp}_{{\mathfrak M}}\in N^w(\left\llbracket\alpha'\right\rrbracket ^{qp}_{{\mathfrak M}})\text{ (by condition (n)*)}\\
& \text{ iff } & \left\llbracket\beta\right\rrbracket _{{\mathfrak M}}\in N^w(\left\llbracket\alpha\right\rrbracket _{{\mathfrak M}})\text{ (by induction hypothesis)}\\
& \text{ iff } & \mathfrak{M},w\vDash B(\alpha,\beta),\\
\end{array}
\]
where $\alpha'=Tr_1(\alpha)$ and $\beta'=Tr_1(\beta)$.
Therefore we can prove that both $Tr_1$ and $Tr_2$ preserve truth value on conditional neighbourhood models.

However we can design models violating the comparison principle of
Figure \ref{fig:Comparison}.  $\mathcal{L}_{QP}$ can differentiate
models where the principle holds from those where it does not, while
$\mathcal{L}_{CN}$ cannot. A \emph{comparison model} $\mathfrak{N}$ is
a triple $(W,\succeq,V)$ where $W$ is a non-empty set of worlds,
${\succeq}\subseteq\mathcal{P}W\times\mathcal{P}W$ is a relation between
propositions, and $V$ is a valuation. Truth definition for
$\mathcal{L}_{QP}$ is with the following key clause:
\[
\mathfrak{N},w\vDash_2\alpha\succcurlyeq\beta\text{ iff }\left\llbracket\alpha\right\rrbracket_{{\mathfrak N}}\succeq\left\llbracket\beta\right\rrbracket_{{\mathfrak N}}.
\]
Furthermore the key clause in the truth condition for
$\mathcal{L}_{CN}$ is given by:
\[
\mathfrak{N},w\vDash_1 B(\alpha,\beta)\text{ iff }\left\llbracket\alpha\wedge\beta\right\rrbracket_{{\mathfrak N}}\succeq\left\llbracket\alpha\wedge\neg\beta\right\rrbracket_{{\mathfrak N}},
\]
which parallels with translation $Tr_1$ at the semantic level.

Let $N_1=(W,\succeq_1,V)$ and $N_2=(W,\succeq_2,V)$ be comparison models such that:
\begin{enumerate}
\item $W=\{\{p,q\},\{p\},\{q\},\emptyset\}$
\item ${\succeq_1}=\{(\{\{p\},\{p,q\}\},\{\{q\},\{p,q\}\}),(\{p\},\{q\})\}$,
\item ${\succeq_2}=\{(\{p\},\{q\})\}$,
\item $w\in V(r)$ iff $r\in w$.
\end{enumerate}
The only difference between $\mathfrak{N}_1$ and $\mathfrak{N}_2$ is
that
$\left\llbracket p\right\rrbracket_{\mathfrak{N}_1}\succeq_1
\left\llbracket q\right\rrbracket_{\mathfrak{N}_1}$ but not
$\left\llbracket p\right\rrbracket_{\mathfrak{N}_2}\succeq_2
\left\llbracket q\right\rrbracket_{\mathfrak{N}_2}$.  Thus $N_2$
violates the \textbf{comparison principle}, and
$\neg(p\succcurlyeq q)\wedge(p\wedge\neg q\succcurlyeq\neg p\wedge q)$
is valid on $\mathfrak{N}_2$ but not on $\mathfrak{N}_1$. However we can
prove that $\mathfrak{N}_1$ and $\mathfrak{N}_2$ satisfy the same set of
$\mathcal{L}_{CN}$-formulas. The crucial fact is to observe that we
only use the comparison relation for disjoint propositions for
$\vDash_1$. For instance
$B(p\leftrightarrow\neg q,p)=Tr_2(p\succcurlyeq q)=Tr_2(p\wedge\neg
q\succcurlyeq p\wedge\neg q)$ is valid on both $\mathfrak{N}_1$ and
$\mathfrak{N}_2$, because for each $i\in\{1,2\}$,
\[
\begin{array}{ccl}
\mathfrak{N}_i,w\vDash_1 B(p\leftrightarrow\neg q,p)) & \text{ iff } & \left\llbracket (p\leftrightarrow\neg q)\wedge p\right\rrbracket_{{\mathfrak N}_i}\succeq_i\\
& & \left\llbracket(p\leftrightarrow\neg q)\wedge\neg p\right\rrbracket_{{\mathfrak N}_i}\\
& \text{ iff } & \left\llbracket p\wedge\neg q\right\rrbracket_{{\mathfrak N}_i}\succeq_i\left\llbracket\neg p\wedge q\right\rrbracket_{{\mathfrak N}_i}\\
& \text{ iff } & \{p\}\succeq_i\{q\}, \text{ which holds for}\\
& & \text{either }i\in\{1,2\}.
\end{array}
\]
Therefore $\mathcal{L}_{QP}$ is more expressive than $\mathcal{L}_{CN}$. We conclude $\mathcal{L}_{CN}$ as a core logic for conditional belief as willingness to bet.

\section{Public Announcement for Conditional Neighborhood Models}\label{section:pub}

Public announcement update for weight models parallels Bayesian update
in probability theory. Public announcement update for probabilistic
logic was first treated in \cite{Kooi03:kcac}, and more complicated
probabilistic updates were discussed in \cite{Benthem2003:cpmul} and
\cite{BenGerKoo09:duwp}. As was mentioned in the introduction, public
announcement updates may destroy reasonable neighbourhood conditions.
The good news is that conditional neighbourhood models are more well
behaved.  We propose two ways of public announcement updating:
deleting points and cutting links; and show reduction axioms for
either of them. These shows that our neighbourhood conditions are some
core principles that preserved by Bayesian update.

\subsection{Deleting Points}
The first approach is the usual one for public announcement update,
which is restricting the domain to $\phi$-worlds after announcing
$\phi$. It assumes that only facts (true propositions at the current
world) can be publicly announced. Public announcements create common 
knowledge, but it need not be the case that a fact that gets announced
becomes true after the update; Moore sentences are a well-known exception.

The language $\mathcal{L}_{PC}$ is the result of extending our base language $\mathcal{L}_{CN}$
with a public announcement operator: 
\[
\phi::=\top\mid p\mid\neg\phi\mid\left(\phi\wedge\phi\right)\mid B_{a}(\phi,\phi)\mid[\phi]\phi
\]
If $\mathfrak{M}=(W,N,V)$ is a conditional neighborhood model, $\sim$ is the induced epistemic relation, 
and $\phi$ is a formula of the $\mathcal{L}_{PC}$ language, then
$\mathfrak{M}^{\phi}=(W^{\phi},{}^{\phi}N,V^{\phi})$ is given by: 
\begin{itemize}
\item $W^{\phi}=\{w\in W\mid\mathfrak{M},w\vDash\phi\}$ 
\item $w\sim_{a}^{\phi}u$ iff $\mathfrak{M},w\vDash\phi$, $\mathfrak{M},u\vDash\phi$
and $w\sim_{a}u$ 
\item $^{\phi}N_{a}^{w}(X)=\left\{
\begin{array}{ll}
N_{a}^{w}(X) & \text{if }X\subseteq W^{\phi}\text{ and }w\in W^{\phi}\\
\text{undefined} & \text{otherwise}
\end{array}\right.$
\item $V^{\phi}(p)=V(p)\cap W^{\phi}$
\end{itemize}

\begin{example} \label{example:lottery2} 
  As an example, consider Alice's lottery situation again.  Alice
  knows there are 10,000 lottery tickets numbered 0000 through 9999.
  Alice believes ticket $t$ is winning (and buys it).  Let $n$
  represent the world where ticket $n$ is winning. Then Alice's belief
  is given by a conditional neighborhood model with
  $N^w_a (X) = \{ Y \subseteq X \mid t \in Y \}$ if
  $t \in X$, and
  $N^w_a (X) = \{ Y \subseteq X \mid |Y| > \frac 1 2 | X | \}$ if
  $t \notin X$.  Now Alice receives the information that $v \neq t$ is
  the winning ticket.  Then $v = w$, the updated model has universe
  $\{ v \}$, and Alice updates her belief to $N'$ with
  ${N'}^v_a (\{ v \}) = \{ \{ v \} \}$. The updated model satisfies
  the conditions for a conditional neighborhood model.
\end{example}

\begin{definition}
  Semantics for $\vDash_{PC}$: let $\mathfrak{M}=\left(W,N,V\right)$ be
  a conditional neighborhood model, let $w\in W$.
\[
\mathfrak{M},w\vDash_{PC}[\phi]\psi\text{ iff }\mathfrak{M},w\vDash_{PC}\phi\text{ implies }\mathfrak{M}^{\phi},w\vDash_{PC}\psi.
\]
\end{definition}

A complete calculus for $\vDash_{PC}$ consists of the calculus for CN,
plus the usual Reduction Axioms of public announcement update for
boolean cases and the following key Reduction Axiom (call this system
PC):
\begin{itemize}
\item $[\phi]B_{a}(\psi,\chi)\leftrightarrow(\phi\to B_{a}(\phi\wedge[\phi]\psi,[\phi]\chi))$
\end{itemize}

In Appendix \ref{Appendix:CompPCandPC'} we prove that the \emph{PC} Calculus is sound and complete w.r.t. $\vDash_{PC}$. 

\subsection{Cutting Links}

We can generalize public announcement of facts to public announcement
of truth values.  In announcing the value of $\phi$, it depends on the
truth value of $\phi$ in the actual world whether $\phi$ or
$\neg \phi$ gets announced.

The language $\mathcal{L}_{PC\pm}$ for this kind of update is given by the following BNF: 
\[
\phi::=\top\mid p\mid\neg\phi\mid\left(\phi\wedge\phi\right)\mid B_{a}(\phi,\phi)\mid[\pm\phi]\phi
\]
To capture the intuition of such updates for conditional neighbourhoods,
we use the following mechanism that cuts epistemic relations between
$\phi$-worlds and $\neg\phi$-worlds after announcing $\phi$.

If $\mathfrak{M}=(W,N,V)$ is a conditional neighborhood model, $\sim$ is the induced epistemic relation, 
and $\phi$ is a formula of the $\mathcal{L}_{PC\pm}$ language, then
$\mathfrak{M}^{\pm\phi}=(W{}^{\pm\phi},{}^{\text{\textpm}\phi}N,V^{\pm\phi})$
is given by: 
\begin{itemize}
\item $W^{\pm\phi}=W$ 
\item $\sim^{\pm\phi}=\{(w,v)\in W^{2}\mid w\sim_{a}v\text{ and }\mathfrak{M},w\vDash\phi\text{ iff }\mathfrak{M},v\vDash\phi\}$
\item $^{\pm\phi}N_{a}^{w}(X)=N_{a}^{w}(X\cap[w]_{a}^{\text{\textpm}\phi})$
\item $V^{\pm\phi}=V$
\end{itemize}

\begin{definition}
Semantics for $\vDash_{PC\pm}$: let $\mathfrak{M}=\left(W,N,V\right)$
be a conditional neighborhood model, let $w\in W$.
\[
\mathfrak{M},w\vdash_{PC\pm}[\phi]\psi\text{ iff }\mathfrak{M}^{\pm\phi},w\vdash_{PC\pm}\psi.
\]
\end{definition}

The system PC$\pm$ for $\vdash_{PC\pm}$ consists of the calculus for
CN, plus the following Reduction Axioms:
\begin{itemize}
\item $[\pm\phi]p\leftrightarrow p$
\item $[\pm\phi]\neg\psi\leftrightarrow\neg[\pm\phi]\psi$
\item $[\pm\phi](\psi\wedge\chi)\leftrightarrow[\pm\phi]\psi\wedge[\pm\phi]\chi$
\item $[\pm\phi]B_{a}(\psi,\chi)\leftrightarrow(\phi\to B_{a}(\phi\wedge[\pm\phi]\psi,[\pm\phi]\chi))\wedge(\neg\phi\to B_{a}(\neg\phi\wedge[\pm\phi]\psi,[\pm\phi]\chi))$
\end{itemize}

Also in Appendix \ref{Appendix:CompPCandPC'} we prove that the \emph{PC}$\pm$ Calculus is sound and complete w.r.t. $\vDash_{PC\pm}$. 

\section{Conclusions and Future Work}\label{section:cafw}

In Section \ref{section:int}, we illustrated that public announcement
update may not preserve reasonable neighbourhood conditions. To
overcome this problem, we introduced conditional neighourhood
semantics in Section \ref{section:cns}. We gave an alternative
interpretation for this system in Section \ref{section:wei}, and then
in Section \ref{section:pub} we gave two flavours of public
announcement update for conditional neighborhood semantics. Because
public announcement update for epistemic weight models is basically Bayesian
update in a logical setting, and because the complete calculus CN for
conditional neighourhood models is a subsystem of a complete
probabilistic logic CN$\oplus$(A4) (as shown in Section
\ref{section:wei}), our investigations show that CN can be viewed as a
core logic of rational belief and belief update that is compatible
with the Bayesian picture of inference, but more general. Conditional
neighborhood models generalize epistemic weight models, and this
generalization creates room for modelling ambiguity aversion in belief as
willingness to bet.

In Section \ref{section:wei} we have shown that $\mathcal{L}_{CN}$ has
weaker expressive power than $\mathcal{L}_{QP}$. Our conditional
neighourhood semantics for $\mathcal{L}_{CN}$ allows us to develop a
reasoning system CN that is not yet committed to a probabilistic
numerical interpretation of belief.  This might be a convenient
starting point for further investigation of counterfactual
reasoning. In Section \ref{section:cns} we have assumed that
conditioning with information that contradicts an agent's knowledge
will cause the agent to refrain believing anything, but in future work
we may relax this constraint, by allowing to visit knowledge cells
other than the current one when a neighbourhood function is
conditioned with propositions that conflict with current knowledge. A
naive way to do so is to incorporate the selection function $f$ for
counterfactuals proposed by Stalnaker in \cite{stalnaker1968theory} in
our framework. When a proposition $X$ is disjoint with agent $a$'s
current knowledge cell $[w]_a$, the selection function $f$ would guide
us to an $X$-world $u=f_a(X,w)$, and then let the neighbourhood
$N_a^w(X)$ be $N_a^u(X)$, which is the neighourhood conditioned by $X$
at $[u]_a$.

While subjective conditional beliefs given by neighbourhood functions
suggest how agents' beliefs would change by public announcements,
further updates like public lies and recovery from lies may allow us
to represent further details of agents' beliefs in an objective way.
Here, recovery is to free agents from the influence of lies that were
accepted as true in the past.  These two kinds of updates give us
powerful tools to test what a agent would believe after providing each
possible piece of information, which in turn would inform us the
agent's conditional beliefs or subjective probability. It is future
work to compare and combine these two approaches: subjective
conditional beliefs informative for belief update and objective belief
changes to conditional beliefs.

Neighourhood structures have also been used to describe the pieces of
evidence that agents accept
(\cite{van2011dynamic},\cite{van2012evidence},\cite{van2014evidence}). In
this approach, each proposition in a neighbourhood is interpreted as a
piece of evidence, instead of a belief; and because evidences are
usually acquired in various situations, such evidence neighbourhoods
may have contradictory propositions. Furthermore, beliefs are
generated from certain consistent subsets of the evidence
neighbourhood. Even though evidence models and conditional
neighbourhood models provide different perspectives on belief, we may
be able to combine them in future.  In one direction, conditional
beliefs or even subjective probability could be generated from certain
evidence models, while the way evidence is involved in belief
formation may provide information about the strengths of the resulting
beliefs.  In the other direction, our conditional beliefs might 
serve as prior knowledge for specifying the credence of evidence. 

\bibliographystyle{eptcs}
\bibliography{CBKP}

\appendix

\section{Alternative Definition of Conditional Neighborhood Models} 
\label{appendix:alternativeDef}

\begin{lemma}
\label{lem:implied-conditions}Let ${\mathfrak M}=\left(W,N,V\right)$
be a conditional neighborhood model. Then ${\mathfrak M}$ satisfies the
following conditions for any $a\in Ag$, $x\in W$ and $X\subseteq W$:
\begin{description}
\item [{(m)}] $\forall Y\subseteq Z\subseteq X\cap[w]_{a}:$ if $Y\in N_{a}^{w}(X)$,
then $Z\in N_{a}^{w}(X)$;
\item [{(no-emptyset)}] $\emptyset\notin N_{a}^{w}(X)$;
\item [{(n){*}}] if $X\cap[w]_{a}\neq\emptyset$, then $X\cap[w]_{a}\in N_{a}^{w}(X)$;
\item [{($\emptyset$)}] if $X\cap[w]_{a}=\emptyset$, then $N_{a}^{w}(X)=\emptyset$.
\end{description}
\end{lemma}
\begin{proof}
Let $a\in Ag$, $x\in W$, $X\subseteq W$ and $X'=X\cap[w]_{a}$.

First consider (m). Let $Y\subseteq Z\subseteq X'$ and $Y\in N_{a}^{w}(X)$.
Suppose for contradiction $Z\notin N_{a}^{w}(X)$. Then $Y\neq Z$,
and hence $Y\subsetneq Z$, which implies $X'-Z\subsetneq X'-Y$.
It follows, from (sc), that $X'-Y\in N_{a}^{w}(X)$, contrary to $Y\in N_{a}^{w}(X)$
and (d).

Second consider (no-emptyset). Suppose for contradiction that $\emptyset\in N_{a}^{w}(X)$.
If $X'=\emptyset$, then by (c) $N_{a}^{w}(X)=\{\emptyset\}$; but
by (d), since $\emptyset\in N_{a}^{w}(X)$, 
\[
\emptyset=X'-\emptyset\notin N_{a}^{w}(X).
\]
Contradiction! If $X'\neq\emptyset$, then since $\emptyset\subseteq X'$,
by (m) we have $X'\in N_{a}^{w}(X)$; but because $\emptyset\in N_{a}^{w}(X)$,
by (d) $X'=X'-\emptyset$ should not in $N_{a}^{w}(X)$, contradiction. 

Third for (n){*}. Suppose $X'\neq\emptyset$ and for contradiction
that $X'\notin N_{a}^{w}(X)$. Then $X'-\emptyset\notin N_{a}^{w}(X)$,
and by (sc) 
\[
\emptyset\subsetneq X'\in N_{a}^{w}(X),
\]
a contradiction!

Last for ($\emptyset$). Suppose $X'=\emptyset$. Then by (c), for
all $Y\in N_{a}^{w}(X)$, $Y=\emptyset$. Because of (no-emptyset),
we have $\emptyset\notin N_{a}^{X}(w)$. Therefore $N_{a}^{w}(X)=\emptyset$.
\end{proof}

Note that in Definition \ref{def:NeighborModel} the equivalence relation 
for knowledge is derived from the conditional neighborhood function. Here 
we will show that there is an equivalent definition that takes epistemic 
equivalences as primary. 

\begin{definition}
\label{def:CNMwithEqui}Let $Ag$ be a finite set of agents. A conditional
neighborhood model{*} ${\mathfrak M}$ is a tuple $\left(W,\sim,N,V\right)$
where:
\begin{itemize}
\item $W$ is a non-empty set of worlds;
\item $\sim$ assigns to each $a\in\Ag$ an equivalence relation $\sim_{a}$ on $W$, and we use $\left[w\right]_{a}$
for the $\sim_{a}$ class of $w$;
\item $N$ assigns to each $a\in\Ag$ a function $N_{a}$ that assigns to every world $w\in W$ and set
of worlds $X\subseteq W$ a collection $N_{a}^{w}(X)$ of sets of
worlds\textendash each such set called a neighborhood of $X$\textendash subject
to the following conditions:
\begin{lyxlist}{00.00.0000}
\item [{(c)}] $\forall Y\in N_{a}^{w}(X):Y\subseteq X\cap[w]_{a}$;
\item [{(a)}] $\forall v\in\left[w\right]_{a}:N_{a}^{w}(X)=N_{a}^{v}(X)$;
\item [{(d)}] $\forall Y\in N_{a}^{w}(X)$, $X\cap[w]_{a}-Y\notin N_{a}^{w}(X)$;
\item [{(sc)}] $\forall Y,Z\subseteq X\cap[w]_{a}:$ if $X-Y\notin N_{a}^{w}(X)$
and $Y\subsetneq Z$, then $Z\in N_{a}^{w}(X)$;
\item [{(ec)}] $\forall Y\subseteq W$: if $X\cap[w]_{a}=Y\cap[w]_{a}$,
then $N_{a}^{w}(X)=N_{a}^{w}(Y)$;
\end{lyxlist}
\item $V$ is a valuation.
\end{itemize}
\end{definition}

Note that in this definition we have another condition (a) on neighborhood
functions. This contrasts with Definition \ref{def:NeighborModel}, where 
we already make sure that (a) holds by the way we define $\left[w\right]_{a}$. 

In Definition \ref{def:CNMwithEqui}, however, $\left[w\right]_{a}$ is defined in terms of 
$\sim_{a}$, which is simply an equivalence relation that does not come with a guarantee
for (a). 

The following proposition assures us that the two approaches are equivalent. 

\begin{proposition}
\label{prop:ModelWithEqvi}Let ${\mathfrak M}=\left(W,\sim,N,V\right)$
be a conditional neighborhood model{*}, let $a\in Ag$ and let $R_{a}\subseteq W\times W$
be defined as follows:
\begin{itemize}
\item $\forall w,v\in W,wR_{a}v$ iff $\forall X\subseteq W$, $N_{a}^{w}(X)=N_{a}^{v}(X)$.
\end{itemize}
Then $\sim_{a}=R_{a}$, and $\left(W,N,V\right)$ is a conditional
neighborhood model.
\end{proposition}
\begin{proof}
Let $w,v\in W$. Suppose $w\sim_{a}v$. Then by (a), we know that
for each $X\subseteq W$, $N_{a}^{w}(X)=N_{a}^{v}(X)$, which implies
$(w,v)\in R_{a}$.

Suppose it is not the case that $w\sim_{a}v$. Then $[w]_{a}\cap[v]_{a}=\emptyset$.
Similar to the proofs in Lemma \ref{lem:implied-conditions}, we can
prove (n){*} and ($\emptyset$) for ${\mathfrak M}$, and hence we have
$[w]_{a}\in N_{a}^{w}([w]_{a})$, and by ($\emptyset$), $[w]_{a}\notin N_{a}^{v}([w]_{a})$.
It follows that $N_{a}^{w}([w]_{a})\neq N_{a}^{v}([w]_{a})$, which
implies $(w,v)\notin R_{a}$.

Therefore $[w]_{a}=\{v\in W\mid\forall X\subseteq W,N_{a}^{w}(X)=N_{a}^{v}(X)\}$.
We can check that $\left(W,N,V\right)$ satisfies all the conditions
in Definition \ref{def:NeighborModel}, which implies $\left(W,N,V\right)$
is a conditional neighborhood model.
\end{proof}

\section{Completeness of CN} 
\label{appendix:CompCNproof}

In this section we prove Theorem \ref{thm:CompCN}. As a first step in the creation of a canonical model, we define formula closures. 

\begin{definition}
Let $\alpha\in\mathcal{L}_{CN}$ be any $\mathit{CN}$-consistent
formula, i.e., $\nvdash_{\mathit{CN}}\neg\alpha$. The basic closure
of $\alpha$, denoted as $\Phi(\alpha)$, is the smallest set of formula
$\Gamma$ such that:
\begin{itemize}
\item if $\phi$ is a sub-formula of $\alpha\wedge\top$, then $\phi\in\Gamma$;
\item if $\phi\in\Gamma$ and $\phi$ is not a negation, then $\neg\phi\in\Gamma$;
\item if $\phi,\psi\in\Gamma$, then $\phi\wedge\psi\in\Gamma$.
\end{itemize}
Let $\Phi^{B}(\alpha)$ be the smallest extension of $\Phi(\alpha)$
such that
\begin{itemize}
\item if $\phi,\psi\in\Phi(\alpha)$, then $B_{a}(\phi,\psi)\in\Phi^{B}(\alpha)$
for each $a\in\Ag$;
\item if $\phi\in\Phi^{B}(\alpha)$ and $\phi$ is not a negation, then $\neg\phi\in\Phi^{B}(\alpha)$;
\item if $\phi,\psi\in\Phi^{B}(\alpha)$, then $\phi\wedge\psi\in\Phi^{B}(\alpha)$.
\end{itemize}
\end{definition}
Now we define the canonical model of $\alpha$. Note that we use
maximal consistent subsets of $\Phi^{B}(\alpha)$ instead of maximal
consistent sets because we want to make sure our model is in some
sense differentiable.  Furthermore, for each maximal consistent subset
of $\Phi^{B}(\alpha)$, we duplicate it $\Omega$ number of times, i.e.,
each possible world is a maximal consistent subset $\mathbf{w}$ of
$\Phi^{B}(\alpha)$ indexed by a number $i\in\Omega$, namely
$\mathbf{w}_i$. In this way, we can define an equivalence relation
$\sim_{a}$ with the right properties in our canonical model.
\begin{definition}
\label{def:canonicalModel}A canonical conditional neighborhood
model $\mathfrak{M}_{\alpha}$ of $\alpha$ is a tuple $(W,\sim,N,V)$
where:
\begin{itemize}
\item $W=\{\mathbf{w}\subseteq\Phi^{B}(\alpha)\mid\text{\ensuremath{\mathbf{w}}}\text{ is a maximal consistent subset of }\Phi^{B}(\alpha)\}\times\Omega$.
\item for each $a\in\Ag$, $\sim_{a}$ is an equivalence relation on $W$
such that $\forall\mathbf{w}_{i},\mathbf{v}_{j}\in W$:
\begin{enumerate}
\item \label{enu:SameBeliefSim}if $\mathbf{w}_{i}\sim_{a}\mathbf{v}_{j}$,
then $\forall\phi,\psi\in\Phi(\alpha),\mathbf{w}\vdash_{CN} B_{a}(\phi,\psi)\text{ iff }\mathbf{v}\vdash_{CN} B_{a}(\phi,\psi)$;
\item \label{enu:distinctSim}if $\mathbf{w}_{i}\sim_{a}\mathbf{v}_{j}$
and $\mathbf{w}\cap\Phi(\alpha)=\mathbf{v}\cap\Phi(\alpha)$, then
$\mathbf{w}_{i}=\mathbf{v}_{j}$;
\item \label{enu:existSim}if $\forall\phi,\psi\in\Phi(\alpha),\mathbf{w}\vdash_{CN} B_{a}(\phi,\psi)\text{ iff }\mathbf{v}\vdash_{CN} B_{a}(\phi,\psi)$,
then there is a $\mathbf{u}_{l}\in W$ such that $\mathbf{v}\cap\Phi(\alpha)=\mathbf{u}\cap\Phi(\alpha)$
and $\mathbf{w}_{i}\sim_{a}\mathbf{u}_{l}$.
\end{enumerate}
\item $N_{a}^{\mathbf{w}}(X):=\{\{\mathbf{v}_{j}\in X\cap\left[\mathbf{w}_{i}\right]_{a}\mid\mathbf{v}\vdash_{CN} \psi\}\mid\phi\in\Phi(\alpha),\mathbf{w}_{i}\vdash_{CN} B_{a}(\left\Vert X\right\Vert _{a}^{\mathbf{w}_{i}},\psi)\}$,
where $\left[\mathbf{w}_{i}\right]_{a}=\{\mathbf{v}_{j}\in W\mid\mathbf{w}_{i}\sim_{a}\mathbf{v}_{j}\}$
and $\left\Vert X\right\Vert _{a}^{\mathbf{w}_{i}}$ is the characteristic
formula for $X$ w.r.t. $\left[\mathbf{w_{i}}\right]_{a}$, i.e.,
$\forall\mathbf{w}_{i}\in\left[\mathbf{w}\right]_{a}$, $\left\Vert X\right\Vert _{a}^{\mathbf{w}_{i}}\in\mathbf{w}$
iff $\mathbf{w}_{i}\in X$.
\item $\mathbf{w}_{i}\in V(p)$ iff $p\in\mathbf{w}$.
\end{itemize}
\end{definition}
%
Note that canonical conditional neighborhood models have equivalence relations $\sim_{a}$, 
unlike conditional neighborhood models. 
However, this is not a problem, because by Proposition \ref{prop:ModelWithEqvi},
conditional neighborhood models with and without $\sim_{a}$ relations are 
essentially the same.

Also note that $\Phi(\alpha)$ is finite up to logical equivalence, and
because $\Ag$ is finite, $\Phi^{B}(\alpha)$ is finite up to logical
equivalence as well. It follows that by Condition (\ref{enu:distinctSim})
for $\sim_{a}$ each $\left[\mathbf{w}_{i}\right]_{a}$ is finite.

Because Condition (\ref{enu:SameBeliefSim}) for $\sim_{a}$,
$\left[\mathbf{w}_{i}\right]_{a}$ is differentiable w.r.t. $\Phi(\alpha)$
in the sense that for each subset $X\subseteq\left[\mathbf{w}_{i}\right]_{a}$,
there is a characteristic formula $\phi\in\Phi(\alpha)$ such that
$\forall\mathbf{v}_{j}\in\left[\mathbf{w}_{i}\right]_{a}$, $\phi\in\mathbf{v}_{j}$
iff $\mathbf{v}_{j}\in X$, and we use $\left\Vert X\right\Vert _{a}^{\mathbf{v}_{j}}$
for such characteristic formula.
\begin{lemma}
A canonical conditional neighborhood model $\mathfrak{M}_{\alpha}$
of $\alpha$ exists, given that $\alpha$ is $\mathit{CN}$-consistent.
\end{lemma}
\begin{proof}
We only need to prove $\sim_{a}$ exists for each $a\in Ag$. Let $a\in Ag$,
let $MCS$ be the set of maximal consistent subsets of $\Phi^{B}(\alpha)$,
let $W=MCS\times\Omega$ and let ${\approx}\subseteq MCS\times MCS$
be the relation such that for all $\mathbf{w},\mathbf{v}\in MCS$,
$\mathbf{w}\approx\mathbf{v}$ iff:
\begin{itemize}
\item $\mathbf{w}\cap\Phi(\alpha)=\mathbf{v}\cap\Phi(\alpha)$,
\item $\forall\phi,\psi\in\Phi(\alpha),\mathbf{w}\vdash_{CN} B_{a}(\phi,\psi)\text{ iff }\mathbf{v}\vdash_{CN} B_{a}(\phi,\psi)$.
\end{itemize}
It is easy to check that $\approx$ is an equivalence relation, and
because $MCS$ is finite, $\left\llbracket \mathbf{w}\right\rrbracket =\{\mathbf{v}\mid\mathbf{w}\approx\mathbf{v}\}$
is also finite. It follows that $\left\llbracket \mathbf{w}\right\rrbracket \times\Omega$
is enumerable. Notice that $\{\left\llbracket \mathbf{v}\right\rrbracket \times\Omega\mid\mathbf{v}\in MCS\}$
is a partition of $W$. For each $\left\llbracket \mathbf{w}\right\rrbracket \times\Omega\in\{\left\llbracket \mathbf{v}\right\rrbracket \times\Omega\mid\mathbf{v}\in MCS\}$,
let $w_{0},w_{1},\ldots$ be an enumeration of $\left\llbracket \mathbf{w}\right\rrbracket \times\Omega$.

Now we define $\sim_{a}$. Let ${\sim_{a}}\subseteq W\times W$ be the
relation such that for all $\mathbf{w}_{i},\mathbf{v}_{j}\in W$,
$\mathbf{w}_{i}\sim_{a}\mathbf{v}_{j}$ iff
\begin{itemize}
\item $\forall\phi,\psi\in\Phi(\alpha),\mathbf{w}\vdash_{CN} B_{a}(\phi,\psi)\text{ iff }\mathbf{v}\vdash_{CN} B_{a}(\phi,\psi)$,
\item $\mathbf{w}_{i}$ and $\mathbf{v}_{j}$ are both the $n$-th element
in enumerations of $\left\llbracket \mathbf{w}\right\rrbracket \times\Omega$
and $\left\llbracket \mathbf{v}\right\rrbracket \times\Omega$ respectively,
i.e., $\mathbf{w}_{i}=w_{n}$ and $\mathbf{v}_{j}=v_{n}$ for some
$n\in\Omega$.
\end{itemize}
To check that such $\sim_{a}$ is a desired equivalence relation,
first it is easy to verify that it satisfies Condition (\ref{enu:SameBeliefSim})
in Definition \ref{def:canonicalModel}. 

Now consider Condition (\ref{enu:distinctSim}) in Definition \ref{def:canonicalModel}.
Suppose $\mathbf{w}_{i}\sim_{a}\mathbf{v}_{j}$ and $\mathbf{w}\cap\Phi(\alpha)=\mathbf{v}\cap\Phi(\alpha)$.
Then we have $\left\llbracket \mathbf{w}\right\rrbracket =\left\llbracket \mathbf{v}\right\rrbracket $,
which implies $\left\llbracket \mathbf{w}\right\rrbracket \times\Omega=\left\llbracket \mathbf{v}\right\rrbracket \times\Omega$.
Furthermore, since $\mathbf{w}_{i}$ and $\mathbf{v}_{j}$ are both
the $n$-th element in enumerations of $\left\llbracket \mathbf{w}\right\rrbracket \times\Omega$
and $\left\llbracket \mathbf{v}\right\rrbracket \times\Omega$ respectively
for some $n\in\Omega$, we obtain $\mathbf{w}_{i}=\mathbf{v}_{j}$.

Lastly consider Condition (\ref{enu:distinctSim}) in Definition \ref{def:canonicalModel}.
Let $\mathbf{w}_{i},\mathbf{v}_{j}\in W$ such that 
\[
\forall\phi,\psi\in\Phi(\alpha),\mathbf{w}\vdash_{CN} B_{a}(\phi,\psi)\text{ iff }\mathbf{v}\vdash_{CN} B_{a}(\phi,\psi).
\]
Suppose $\mathbf{w}_{i}$ is the $n$-th element in the enumeration
of $\left\llbracket \mathbf{w}\right\rrbracket \times\Omega$. Then
there is a $\mathbf{u}_{l}\in\left\llbracket \mathbf{v}\right\rrbracket \times\Omega$
such that $\mathbf{u}_{l}$ is the $n$-th element in the enumeration
of $\left\llbracket \mathbf{v}\right\rrbracket \times\Omega$. By
the definition of $\approx$ we have $\mathbf{v}\cap\Phi(\alpha)=\mathbf{u}\cap\Phi(\alpha)$
and 
\[
\forall\phi,\psi\in\Phi(\alpha),\mathbf{v}\vdash_{CN} B_{a}(\phi,\psi)\text{ iff }\mathbf{u}\vdash_{CN} B_{a}(\phi,\psi),
\]
the latter of which implies 
\[
\forall\phi,\psi\in\Phi(\alpha),\mathbf{w}\vdash_{CN} B_{a}(\phi,\psi)\text{ iff }\mathbf{u}\vdash_{CN} B_{a}(\phi,\psi).
\]
Therefore by the definition of $\sim_{a}$ we have $\mathbf{w}_{i}\sim_{a}\mathbf{u}_{l}$
and $\mathbf{v}\cap\Phi(\alpha)=\mathbf{u}\cap\Phi(\alpha)$.
\end{proof}

\begin{lemma}
\label{lem:ConditionEquivalence}Let $\mathfrak{M}_{\alpha}=(W,\sim,N,V)$
be a canonical conditional neighborhood model of a $\mathit{CN}$-consistent
formula $\alpha$, let $a\in Ag$, $\mathbf{w}\in W$, and let $\phi\in\Phi(\alpha)$
such that $\mathbf{v}\vdash_{CN} \phi$ for all $\mathbf{v}\in[\mathbf{w}]_{a}$.
Then $K_{a}\phi\in\mathbf{w}$.
\end{lemma}
\begin{proof}
Because $\phi\in\Phi(\alpha)$, we know that either $K_{a}\phi\in\mathbf{w}$
or $\neg K_{a}\phi\in\mathbf{w}$. Suppose for contradiction that
$\neg K_{a}\phi\in\mathbf{w}$. Then $\mathbf{w}\vdash_{CN} \check{K}_{a}\neg\phi$.
Let 
\begin{itemize}
\item $\mathbf{v}^{-}=\{\neg\phi\}\cup\{B_{a}(\psi,\chi)\in\mathbf{w}\mid\psi,\chi\in\Phi(\alpha)\}\cup\{\neg B_{a}(\psi,\chi)\in\mathbf{w}\mid\psi,\chi\in\Phi(\alpha)\}$.
\end{itemize}
$\mathbf{v}^{-}$ should be consistent, for otherwise there are 
\[
B_{a}(\alpha_{1},\beta_{1}),\ldots,B_{a}(\alpha_{k},\beta_{k}),\neg B_{a}(\gamma_{1},\delta_{1}),\ldots,\neg B_{a}(\gamma_{l},\delta_{l})\in\mathbf{w}
\]
 such that
\[
\vdash_{CN} \bigwedge_{i=1}^{k}B_{a}(\alpha_{i},\beta_{i})\wedge\bigwedge_{i=1}^{l}\neg B_{a}(\gamma_{i},\delta_{i})\to\phi,
\]
which implies, by (Nec-K) and (Dist-K), 
\[
\vdash_{CN} \bigwedge_{i=1}^{k}K_{a}B_{a}(\alpha_{i},\beta_{i})\wedge\bigwedge_{i=1}^{l}K_{a}\neg B_{a}(\gamma_{i},\delta_{i})\to K_{a}\phi,
\]
and then using (5B) and (4B) we have $\mathbf{w}\vdash_{CN} K_{a}\phi$,
contrary to $\mathbf{w}\vdash_{CN} \check{K}_{a}\neg\phi$. It follows that
$\mathbf{v}^{-}$ is consistent, and then by Condition (\ref{enu:existSim})
in Definition \ref{def:canonicalModel}, there is a $\mathbf{v}\in\left[\mathbf{w}\right]_{a}$
such that $\neg\phi\in\mathbf{v}$, contrary to our assumption that
for all $\mathbf{v}\in[\mathbf{w}]_{a}$, $\mathbf{v}\vdash_{CN} \phi$. 

Therefore $\mathbf{w}\vdash_{CN} K_{a}\phi$, i.e., $K_{a}\phi\in\mathbf{w}$.
\end{proof}
\begin{theorem}
Every $\mathit{CN}$-consistent formula $\alpha$ has a conditional
neighborhood model $\mathfrak{M}_{\alpha}$.
\end{theorem}
\begin{proof}
Suppose $\nvdash\neg\alpha$. Let $\mathfrak{M}_{\alpha}=(W,\sim,N,V)$
be a canonical conditional neighborhood model of $\alpha$, and
for all $a\in\Ag$, $\mathbf{w}\in W$ and $X\subseteq W$, let $\left\Vert X\right\Vert _{a}^{w}$
be the characteristic formula for $X$ w.r.t. $\left[\mathbf{w}\right]_{a}$,
i.e., $\forall\mathbf{v}\in\left[\mathbf{w}\right]_{a}$, $\left\Vert X\right\Vert _{a}^{\mathbf{w}}\in\mathbf{v}$
iff $\mathbf{v}\in X$. It suffice to show that $(W,\sim,N,V)$ is
a conditional neighborhood model{*} (see Definition \ref{def:CNMwithEqui}), and then by Proposition \ref{prop:ModelWithEqvi}
we can obtain that $(W,N,V)$ is a conditional neighborhood
model.

Clearly $\sim_{a}$ are equivalence relations. It follows that for
each $\mathbf{v}\in\left[\mathbf{w}\right]_{a}$, $\left[\mathbf{v}\right]_{a}=\left[\mathbf{w}\right]_{a}$,
which implies
\begin{equation}
\forall\mathbf{w},\mathbf{v}\in W\forall X\subseteq W,\mathbf{w}\sim_{a}\mathbf{v}\text{ only if }\left\Vert X\right\Vert _{a}^{\mathbf{w}}=\left\Vert X\right\Vert _{a}^{\mathbf{v}}.\label{eq:CharaFormEq}
\end{equation}
 It remains to show that $N$ satisfies (c), (a), (d), (sc) and (ec).
Let $a\in\Ag$, $\mathbf{w}\in W$ and $X\subseteq W$.

First we consider (c), but it is straightforward by the definition
of $N_{a}^{\mathbf{w}}(X)$.

Second for (a). Consider all $\mathbf{u}\in[\mathbf{w}]_{a}$, $Y,Z\subseteq W$.
$Z\in N_{a}^{\mathbf{w}}(Y)$ iff 
\[
\exists\psi\in\Phi(\alpha),\mathbf{w}\vdash_{CN} B_{a}(\left\Vert Y\right\Vert _{a}^{\mathbf{w}},\psi)\text{ and }Z=\{\mathbf{v}\in Y\cap\left[\mathbf{w}\right]_{a}\mid\mathbf{v}\vdash_{CN} \psi\}
\]
iff by $\left\Vert X\right\Vert _{a}^{\mathbf{w}},\psi\in\Phi(\alpha)$,
Definition \ref{def:canonicalModel}(\ref{enu:SameBeliefSim}) and
(\ref{eq:CharaFormEq})
\[
\exists\psi\in\Phi(\alpha),\mathbf{u}\vdash_{CN} B_{a}(\left\Vert Y\right\Vert _{a}^{\mathbf{u}},\psi)\text{ and }Z=\{\mathbf{v}\in Y\cap\left[\mathbf{u}\right]_{a}\mid\mathbf{v}\vdash_{CN} \psi\}
\]
iff $Z\in N_{a}^{\mathbf{u}}(Y)$.

Third for (d), where we use axioms (D), (N) and (M). Consider any
$Y\in N_{a}^{\mathbf{w}}(X)$. By the definition of $N_{a}^{\mathbf{w}}(X)$,
we have that there is a $\phi\in\Phi(\alpha)$ such that $Y=\{\mathbf{v}\in X\cap\left[\mathbf{w}\right]_{a}\mid\mathbf{v}\vdash_{CN} \phi\}$
and $\mathbf{w}\vdash_{CN} B_{a}(\left\Vert X\right\Vert _{a}^{\mathbf{w}},\phi)$.
Using (D) we can derive that 
\begin{equation}
\mathbf{w}\vdash_{CN} \neg B_{a}(\left\Vert X\right\Vert _{a}^{\mathbf{w}},\neg\phi).\label{eq:connonical(d)}
\end{equation}
Now suppose by contradiction that there is a $\psi\in\Phi$ such that
$X-Y=\{\mathbf{v}\in X\mid\mathbf{v}\vdash_{CN} \psi\}$ and $\mathbf{w}\vdash_{CN} B_{a}(\left\Vert X\right\Vert _{a}^{\mathbf{w}},\psi)$.
Then by (N) we have $\mathbf{w}\vdash_{CN} B_{a}(\left\Vert X\right\Vert _{a}^{\mathbf{w}},\left\Vert X\right\Vert _{a}^{\mathbf{w}}\wedge\psi)$.
Note that for each $\mathbf{v}\in[\mathbf{w}]_{a}$, $\mathbf{v}\vdash_{CN} \left\Vert X\right\Vert _{a}^{\mathbf{w}}\wedge\psi$
only if $\mathbf{v}\in X-Y$ only if $\mathbf{v}\vdash_{CN} \neg\phi$,
i.e., $\mathbf{v}\vdash_{CN} \left\Vert X\right\Vert _{a}^{\mathbf{w}}\wedge\psi\to\neg\phi$.
Thus by Lemma \ref{lem:ConditionEquivalence} and $\left\Vert X\right\Vert _{a}^{\mathbf{w}},\phi,\psi\in\Phi(\alpha)$
we can get that $K_{a}(\left\Vert X\right\Vert _{a}^{\mathbf{w}}\wedge\psi\to\neg\phi)\in\mathbf{w}$,
and then using (M) we obtain $\mathbf{w}\vdash_{CN} B_{a}(\left\Vert X\right\Vert _{a}^{\mathbf{w}},\neg\phi)$,
contrary to (\ref{eq:connonical(d)}).

Then for (sc), we use (T) and (SC). Consider any $Y,Z\subseteq X\cap[\mathbf{w}]_{a}$
such that $X\cap[\mathbf{w}]_{a}-Y\notin N_{a}^{\mathbf{w}}(X)$.
Because $X\cap[\mathbf{w}]_{a}-Y\notin N_{a}^{\mathbf{w}}(X)$, we
have $\mathbf{w}\nvdash B_{a}(\left\Vert X\right\Vert _{a}^{\mathbf{w}},\neg\left\Vert Y\right\Vert _{a}^{\mathbf{w}})$.
Recall that $\left\Vert X\right\Vert _{a}^{\mathbf{w}},\neg\left\Vert Y\right\Vert _{a}^{\mathbf{w}}\in\Phi(\alpha)$,
we know that either $B(\left\Vert X\right\Vert _{a}^{\mathbf{w}},\neg\left\Vert Y\right\Vert _{a}^{\mathbf{w}})\in\mathbf{w}$
or $\neg B(\left\Vert X\right\Vert _{a}^{\mathbf{w}},\neg\left\Vert Y\right\Vert _{a}^{\mathbf{w}})\in\mathbf{w}$.
It follows that 
\begin{equation}
\mathbf{w}\vdash_{CN} \neg B(\left\Vert X\right\Vert _{a}^{\mathbf{w}},\neg\left\Vert Y\right\Vert _{a}^{\mathbf{w}}).\label{eq:cannonical(sc)}
\end{equation}
If $X\cap[\mathbf{w}]_{a}-Y=\emptyset$, then $X\cap[\mathbf{w}]_{a}=Y$,
which implies there is no such $Z$ with $Y\subsetneq Z$; thus (sc)
vacuously holds. Suppose $\mathbf{v}\in X\cap[\mathbf{w}]_{a}-Y$
and $\mathbf{v}\in Z\supsetneq Y$. Then $\mathbf{v}\vdash_{CN} \left\Vert X\right\Vert _{a}^{\mathbf{w}}\wedge\neg\left\Vert Y\right\Vert _{a}^{\mathbf{w}}\wedge\left\Vert Z\right\Vert _{a}^{\mathbf{w}}$.
Using (T) we can obtain that $\mathbf{v}\vdash_{CN} \check{K}_{a}(\left\Vert X\right\Vert _{a}^{\mathbf{w}}\wedge\neg\left\Vert Y\right\Vert _{a}^{\mathbf{w}}\wedge\left\Vert Z\right\Vert _{a}^{\mathbf{w}})$.
By $\mathbf{v}\in[\mathbf{w}]_{a}$ and Condition (\ref{enu:SameBeliefSim})
in Definition \ref{def:canonicalModel}, $\mathbf{w}\vdash_{CN} \check{K}_{a}(\left\Vert X\right\Vert _{a}^{\mathbf{w}}\wedge\neg\left\Vert Y\right\Vert _{a}^{\mathbf{w}}\wedge\left\Vert Z\right\Vert _{a}^{\mathbf{w}})$.
It follows that, using (\ref{eq:cannonical(sc)}) and (SC), $\mathbf{w}\vdash_{CN}  B_{a}(\left\Vert X\right\Vert _{a}^{\mathbf{w}},\left\Vert Y\right\Vert _{a}^{\mathbf{w}}\vee\left\Vert Z\right\Vert _{a}^{\mathbf{w}})$.
Because $Y\subsetneq Z$ and thus $Z=\{\mathbf{v}\in X\cap\left[\mathbf{w}\right]_{a}\mid\mathbf{v}\vdash_{CN} \left\Vert Y\right\Vert _{a}^{\mathbf{w}}\vee\left\Vert Z\right\Vert _{a}^{\mathbf{w}}\}$,
we have $Z\in N_{a}^{\mathbf{w}}(X)$.

Last for (ec). Consider any $Y\subseteq W$ such that $X\cap[\mathbf{w}]_{a}=Y\cap[\mathbf{w}]_{a}$.
Clearly $X\cap[\mathbf{w}]_{a}$ and $Y\cap[\mathbf{w}]_{a}$ have
the same characteristic formula w.r.t. $[\mathbf{w}]_{a}$, i.e.,
$\left\Vert X\right\Vert _{a}^{\mathbf{w}}=\left\Vert Y\right\Vert _{a}^{\mathbf{w}}$.
Let $Z\subseteq W$. $Z\in N_{a}^{\mathbf{w}}(X)$ iff 
\[
\exists\psi\in\Phi(\alpha),\mathbf{w}\vdash_{CN}  B_{a}(\left\Vert X\right\Vert _{a}^{\mathbf{w}},\psi)\text{ and }Z=\{\mathbf{v}\in X\cap\left[\mathbf{w}\right]_{a}\mid\mathbf{v}\vdash_{CN}\psi\}
\]
iff since $\left\Vert X\right\Vert _{a}^{\mathbf{w}}=\left\Vert Y\right\Vert _{a}^{\mathbf{w}}$
and $X\cap[\mathbf{w}]_{a}=Y\cap[\mathbf{w}]_{a}$,
\[
\exists\psi\in\Phi(\alpha),\mathbf{w}\vdash_{CN} B_{a}(\left\Vert Y\right\Vert _{a}^{\mathbf{w}},\psi)\text{ and }Z=\{\mathbf{v}\in Y\cap\left[\mathbf{w}\right]_{a}\mid\mathbf{v}\vdash_{CN} \psi\}
\]
iff $Z\in N_{a}^{\mathbf{w}}(Y)$.
\end{proof}
\begin{lemma}
(Truth Lemma) Let $\alpha\in\mathcal{L}_{CN}$ be a $\mathit{CN}$-consistent
formula and let $\mathfrak{M}_{\alpha}=(W,N,V)$ be a canonical conditional
neighborhood model of $\alpha$ removing equivalence relation $\sim$. Then for all formulas $\phi\in\Phi(\alpha)$
and $\mathbf{w}\in W$, $\mathfrak{M}_{\alpha},\mathbf{w}\vDash_{CN} \phi$
iff $\phi\in\mathbf{w}$.
\end{lemma}
\begin{proof}
We prove by induction on $\phi$. The cases of $\top,p$ and the Boolean
combinations are straightforward. For the case of $B_{a}(\psi,\chi)$.
Let $X$ be any set such that $\{\mathbf{v}\in[\mathbf{w}]_{a}\mid\psi\in\mathbf{v}\}\subseteq X$.
Note that because $B_{a}(\psi,\chi)\in\Phi(\alpha)$, we have $\psi,\chi\in\Phi(\alpha)$.

Note that for each $\mathbf{v}\in[\mathbf{w}]_{a}$, $\psi\in\mathbf{v}$
iff $\mathbf{v}\in X$ iff $\left\Vert X\right\Vert _{a}^{\mathbf{w}}\in\mathbf{v}$.
Thus for each $\mathbf{v}\in[\mathbf{w}]_{a}$, $\mathbf{v}\vdash_{CN} \psi\leftrightarrow\left\Vert X\right\Vert _{a}^{\mathbf{w}}$.
It follows, by Lemma \ref{lem:ConditionEquivalence}
\begin{equation}
K_{a}(\phi\leftrightarrow\left\Vert X\right\Vert _{a}^{\mathbf{w}})\in\mathbf{w}.\label{eq:TL(ec)}
\end{equation}

Also note that $\mathfrak{M}_{\alpha},\mathbf{w}\vDash B_{a}(\psi,\chi)$
iff $\{\mathbf{v}\in\left\llbracket \psi\right\rrbracket _{\text{\ensuremath{\mathfrak{M}_{\alpha}}}}\cap[\mathbf{w}]_{a}\mid\mathfrak{M}_{\alpha},\mathbf{v}\vDash\chi\}\in N_{a}^{\mathbf{w}}(\left\llbracket \psi\right\rrbracket _{\mathfrak{M}_{\alpha}})$
iff (induction hypothesis) $X=\left\llbracket \psi\right\rrbracket _{\text{\ensuremath{\mathfrak{M}}}}$
and $\{\mathbf{v}\in X\cap[\mathbf{w}]_{a}\mid\chi\in\mathbf{v}\}\in N_{a}^{\mathbf{w}}(X)$.
Thus 
\begin{equation}
\mathfrak{M}_{\alpha},\mathbf{w}\vDash B_{a}(\psi,\chi)\text{ iff }\{\mathbf{v}\in X\cap[\mathbf{w}]_{a}\mid\chi\in\mathbf{v}\}\in N_{a}^{\mathbf{w}}(X).\label{eq:TLClaim}
\end{equation}

Suppose $B_{a}(\psi,\chi)\in\mathbf{w}$. Then we have $\mathbf{w}\vdash B_{a}(\psi,\chi)$,
which implies by (\ref{eq:TL(ec)}) and (ec) $\mathbf{w}\vdash B_{a}(\left\Vert X\right\Vert _{a}^{\mathbf{w}},\chi)$.
It follows that $\{\mathbf{v}\in X\cap[\mathbf{w}]_{a}\mid\chi\in\mathbf{v}\}\in N_{a}^{\mathbf{w}}(X)$,
and hence, using (\ref{eq:TLClaim}), $\mathfrak{M}_{\alpha},\mathbf{w}\vDash B_{a}(\psi,\chi)$.

Suppose $\mathfrak{M}_{\alpha},\mathbf{w}\vDash B_{a}(\psi,\chi)$,
and hence by (\ref{eq:TLClaim}) $\{\mathbf{v}\in X\cap[\mathbf{w}]_{a}\mid\chi\in\mathbf{v}\}\in N_{a}^{\mathbf{w}}(X)$.
Then there is a $\chi'\in\Phi(\alpha)$ such that $\{\mathbf{v}\in X\cap[\mathbf{w}]_{a}\mid\chi'\in\mathbf{v}\}=\{\mathbf{v}\in X\cap[\mathbf{w}]_{a}\mid\chi\in\mathbf{v}\}$
and $\mathbf{w}\vdash B_{a}(\left\Vert X\right\Vert _{a}^{\mathbf{w}},\chi')$.
It follows, by Lemma \ref{lem:ConditionEquivalence} that $K_{a}(\chi'\leftrightarrow\chi)\in\mathbf{w}$.
Using (\ref{eq:TL(ec)}), (ec) and (M) we have $\mathbf{w}\vdash B_{a}(\psi,\chi)$.
Recall that $B_{a}(\psi,\chi)\in\Phi(\alpha)$, we can obtain $B_{a}(\psi,\chi)\in\mathbf{w}$.
\end{proof}

\section{Completeness of PC and PC$\pm$}\label{Appendix:CompPCandPC'}

\begin{theorem}
(Soundness) $PC$ is sound w.r.t. $\vDash_{PC}$.
\end{theorem}
\begin{proof}
To illustrate that $[\phi]B_{a}(\psi,\chi)\leftrightarrow(\phi\to B_{a}(\phi\wedge[\phi]\psi,\phi\wedge[\phi]\chi))$
is sound. Consider any conditional neighborhood model $\mathfrak{M}=(W,N,V)$
and any $w\in W$. 

$\mathfrak{M},w\vDash[\phi]B_{a}(\psi,\chi)$, iff $\mathfrak{M},w\vDash\phi$
only if $\mathfrak{M}^{\phi},w\vDash B_{a}(\psi,\chi)$, iff $\mathfrak{M},w\vDash\phi$
only if $\{v\in\left\llbracket \psi\right\rrbracket _{\text{\ensuremath{\mathfrak{M}}}^{\phi}}\cap[w]_{a}^{\phi}\mid\mathfrak{M}^{\phi},w\vDash\chi\}\in M_{a}^{w}(\left\llbracket \psi\right\rrbracket _{\text{\ensuremath{\mathfrak{M}}}^{\phi}})$.

$\mathfrak{M},w\vDash\phi\to B_{a}(\phi\wedge[\phi]\psi,\phi\wedge[\phi]\chi)$
iff $\mathfrak{M},w\vDash\phi$ only if $\mathfrak{M},w\vDash B_{a}(\phi\wedge[\phi]\psi,\phi\wedge[\phi]\chi)$,
iff $\mathfrak{M},w\vDash\phi$ only if $\{v\in\left\llbracket \phi\wedge[\phi]\psi\right\rrbracket _{\text{\ensuremath{\mathfrak{M}}}}\cap[w]_{a}\mid\mathfrak{M},v\vDash\phi\wedge[\phi]\chi\}\in N_{a}^{w}(\left\llbracket \phi\wedge[\phi]\psi\right\rrbracket _{\text{\ensuremath{\mathfrak{M}}}})$.

Let $v\in W$.
\begin{enumerate}
 \item $\mathfrak{M},v\vDash\phi\wedge[\phi]\chi$ iff
$\mathfrak{M},v\vDash\phi$ and $\mathfrak{M},v\vDash\phi$ implies
$\mathfrak{M}^{\phi},v\vDash\chi$ iff $\mathfrak{M}^{\phi},v\vDash\chi$.
 \item $\mathfrak{M},v\vDash\phi\wedge[\phi]\psi$ iff
$\mathfrak{M},v\vDash\phi$ and $\mathfrak{M},v\vDash\psi$ implies
$\mathfrak{M}^{\phi},v\vDash\psi$ iff $\mathfrak{M}^{\phi},v\vDash\psi$.
 \item $\mathfrak{M},w\vDash\phi$ implies $v\in\left\llbracket \phi\wedge[\phi]\psi\right\rrbracket _{\text{\ensuremath{\mathfrak{M}}}}\cap[w]_{a}$
iff $\mathfrak{M},w\vDash\phi$ implies $\mathfrak{M},v\vDash\phi\wedge[\phi]\psi$
and $w\sim_{a}v$ iff $\mathfrak{M}^{\phi},v\vDash\psi$ and $w\sim_{a}^{\phi}v$
iff $v\in\left\llbracket \psi\right\rrbracket _{\text{\ensuremath{\mathfrak{M}}}^{\phi}}\cap[w]_{a}^{\phi}$.
\end{enumerate}

It follows that $\mathfrak{M},w\vDash\phi\to B_{a}(\phi\wedge[\phi]\psi,\phi\wedge[\phi]\chi)$
iff $\mathfrak{M},w\vDash\phi$ implies $\{v\in\left\llbracket \psi\right\rrbracket _{\text{\ensuremath{\mathfrak{M}}}^{\phi}}\cap[w]_{a}^{\phi}
\mid\mathfrak{M}^{\phi},v\vDash\chi\}\in M_{a}^{w}(\left\llbracket \psi\right\rrbracket _{\text{\ensuremath{\mathfrak{M}}}^{\phi}})$, iff $\mathfrak{M},w\vDash[\phi]B_{a}(\psi,\chi)$, and this completes our proof.
\end{proof}

\begin{theorem}
The PC Calculus is complete. 
\end{theorem}
\begin{proof}
This follows directly from the completeness of CN, plus the fact that the axioms for 
public announcement update are reduction axioms: we can compile out the update operators
to reduce PC to CN. 
\end{proof}

\begin{theorem}
(Soundness) PC$\pm$ is sound w.r.t. $\vDash_{PC\pm}$.
\end{theorem}
\begin{proof}
To illustrate that $[\pm\phi]B_{a}(\psi,\chi)\leftrightarrow(\phi\to B_{a}(\phi\wedge[\pm\phi]\psi,[\pm\phi]\chi))\wedge(\neg\phi\to B_{a}(\neg\phi\wedge[\pm\phi]\psi,[\pm\phi]\chi))$
is sound. Consider any conditional neighborhood model $\mathfrak{M}=(W,N,V)$
and any $w\in W$. We consider two cases: $\mathfrak{M},w\vDash\phi$ or $\mathfrak{M},w\vDash\neg\phi$.

First suppose $\mathfrak{M},w\vDash\phi$. Then $\mathfrak{M},w\vDash[\pm\phi]B_{a}(\psi,\chi)$, iff $\mathfrak{M}^{\pm\phi},w\vDash B_{a}(\psi,\chi)$, iff $\{v\in\left\llbracket \psi\right\rrbracket _{\text{\ensuremath{\mathfrak{M}}}^{\pm\phi}}\cap[w]_{a}^{\pm\phi}\mid\mathfrak{M}^{\pm\phi},w\vDash\chi\}\in M_{a}^{w}(\left\llbracket \psi\right\rrbracket _{\text{\ensuremath{\mathfrak{M}}}^{\pm\phi}})$. 

Furthermore, $\mathfrak{M},w\vDash(\phi\to B_{a}(\phi\wedge[\pm\phi]\psi,[\pm\phi]\chi))\wedge(\neg\phi\to B_{a}(\neg\phi\wedge[\pm\phi]\psi,[\pm\phi]\chi))$ iff (because $\mathfrak{M},w\vDash\phi$,) $\mathfrak{M},w\vDash B_{a}(\phi\wedge[\pm\phi]\psi,[\pm\phi]\chi)$ iff 
\[
\{v\in\left\llbracket \phi\wedge[\pm\phi]\psi\right\rrbracket _{\text{\ensuremath{\mathfrak{M}}}}\cap[w]_{a}\mid\mathfrak{M},v\vDash\phi\wedge[\pm\phi]\chi\}\in N_{a}^{w}(\left\llbracket \phi\wedge[\pm\phi]\psi\right\rrbracket _{\text{\ensuremath{\mathfrak{M}}}}).
\]

Let $v\in W$.
\begin{enumerate}
 \item $\mathfrak{M},v\vDash\phi\wedge[\pm\phi]\chi$ iff
$\mathfrak{M},v\vDash\phi$ and $\mathfrak{M},v\vDash\phi$ implies
$\mathfrak{M}^{\pm\phi},v\vDash\chi$ iff $\mathfrak{M}^{\pm\phi},v\vDash\chi$.
 \item $\mathfrak{M},v\vDash\phi\wedge[\pm\phi]\psi$ iff
$\mathfrak{M},v\vDash\phi$ and $\mathfrak{M},v\vDash\psi$ implies
$\mathfrak{M}^{\pm\phi},v\vDash\psi$ iff $\mathfrak{M}^{\pm\phi},v\vDash\psi$.
 \item $\mathfrak{M},w\vDash\phi$ implies $v\in\left\llbracket \phi\wedge[\pm\phi]\psi\right\rrbracket _{\text{\ensuremath{\mathfrak{M}}}}\cap[w]_{a}$
iff $\mathfrak{M},w\vDash\phi$ implies $\mathfrak{M},v\vDash\phi\wedge[\pm\phi]\psi$
and $w\sim_{a}v$ iff $\mathfrak{M}^{\pm\phi},v\vDash\psi$ and $w\sim_{a}^{\pm\phi}v$
iff $v\in\left\llbracket \psi\right\rrbracket _{\text{\ensuremath{\mathfrak{M}}}^{\pm\phi}}\cap[w]_{a}^{\pm\phi}$.
\end{enumerate}

It follows that $\mathfrak{M},w\vDash(\phi\to B_{a}(\phi\wedge[\pm\phi]\psi,[\pm\phi]\chi))\wedge(\neg\phi\to B_{a}(\neg\phi\wedge[\pm\phi]\psi,[\pm\phi]\chi))$
iff $\{v\in\left\llbracket \psi\right\rrbracket _{\text{\ensuremath{\mathfrak{M}}}^{\pm\phi}}\cap[w]_{a}^{\pm\phi}
\mid\mathfrak{M}^{\pm\phi},v\vDash\chi\}\in M_{a}^{w}(\left\llbracket \psi\right\rrbracket _{\text{\ensuremath{\mathfrak{M}}}^{\pm\phi}})$, iff $\mathfrak{M},w\vDash[\pm\phi]B_{a}(\psi,\chi)$.

The proof of the case that $\mathfrak{M},w\vDash\neg\phi$ is similar to that of the first case, and we omit the proof.  
\end{proof}
\begin{theorem}
The PC$\pm$ calculus is complete. 
\end{theorem}
\begin{proof}
Again, this follows directly from the completeness of CN, plus the fact that the axioms for 
public announcement update are reduction axioms: we can compile out the update operators
to reduce PC$\pm$ to CN. 
\end{proof}

\end{document}